\documentclass{article}
\usepackage{amsthm}
\usepackage{amsmath}
\usepackage{graphicx}
\newtheorem{theorem}{Theorem}[section]
\newtheorem{corollary}{Corollary}[theorem]
\newtheorem{lemma}[theorem]{Lemma}
\theoremstyle{definition}

\usepackage{arxiv}
\usepackage{xcolor}
\usepackage[utf8]{inputenc} % allow utf-8 input
\usepackage[T1]{fontenc}    % use 8-bit T1 fonts
\usepackage{hyperref}       % hyperlinks
\usepackage{url}            % simple URL typesetting
\usepackage{booktabs}       % professional-quality tables
\usepackage{amsfonts}       % blackboard math symbols
\usepackage{nicefrac}       % compact symbols for 1/2, etc.
\usepackage{microtype}      % microtypography
\usepackage{enumitem}
\numberwithin{equation}{section}
\DeclareMathOperator{\Tr}{Tr}
\newcommand{\R}{\mathbb{R}}

\newcommand{\p}{\textcolor{white}{....}}
\title{Interpolation of Operators With Trace Inequalities Related To The Positive Weighted Geometric Mean}

\author{
Victoria M.~Chayes\\
Department of Mathematics\\
Rutgers University\\
Piscataway, NJ 08854 \\
\texttt{vc362@math.rutgers.edu} \\
}

\begin{document}

\maketitle

\begin{abstract}
There are various generalizations of the geometric mean $(a,b)\mapsto a^{1/2}b^{1/2}$ for $a,b\in \mathbb{R}^+$ to positive matrices, and we consider the standard positive geometric mean $(X,Y)\mapsto X^{1/2}(X^{-1/2}YX^{-1/2})^{1/2}X^{1/2}$. Much research in recent years has been devoted to relating the weighted version of this mean $X\#_{t}Y:=X^{1/2}(X^{-1/2}YX^{-1/2})^{t}X^{1/2}$ for $t\in [0, 1]$ with operators $e^{(1-t)X+tY}$ and $e^{(1-t)X/2}e^{tY}e^{(1-t)X/2}$ in Golden-Thompson-like inequalities. These inequalities are of interest to mathematical physicists for their relationship to quantum entropy, relative quantum entropy, and R\'{e}nyi divergences. However, the weighted mean is well-defined for the full range of $t\in\mathbb{R}$. In this paper we examine the value of $|||e^H\#_te^K|||$ and variations thereof in comparison to $|||e^{(1-t)H+tK}|||$ and $|||e^{(1-t)H}e^{tK}|||$ for any unitarily invariant norm $|||\cdot|||$ and in particular the trace norm, creating for the first time the full picture of interpolation of the weighted geometric mean with the Golden-Thompson Inequality. We expand inequalities known for $|||(e^{rH}\#_te^{rK})^{1/r}|||$ with $r>0$, $t\in [0,1]$ to the entire real line, and comment on how the exterior inequalities can be used to provide elegant proofs of the known inequalities for $t\in [0,1]$. We also characterize the equality cases for strictly increasing unitarily invariant norms.
\end{abstract}

\section{Introduction}

\p Consider the following functions on positive definite complex matrices: \begin{align}
	&(X,Y)\mapsto e^{\log(X)+\log(Y)}, \\
	&(X,Y)\mapsto e^{\log(X)/2}e^{\log(Y)}e^{\log(X)/2}=X^{1/2}YX^{1/2}, \\
	&(X,Y)\mapsto X^{1/2}(X^{-1/2}YX^{-1/2})^{1/2}X^{1/2}.
\end{align}
Each can be seen as an operator analogue of the geometric mean function $(a,b)\mapsto a^{1/2}b^{1/2}$ defined on $[0,\infty)\times[0,\infty)$. In fact, the third is known as the \textit{geometric mean} of positive matrices, first introduced by Pusz and Woronowicz \cite{PuszWormean} in 1975 as a way of generalizing $\sqrt{xy}$ to sesquilinear forms. 

\p Mathematicians and mathematical physicists will recognize the first two operator functions from their placement in the Golden-Thompson Inequality, proven by Golden for non-negative definite matrices in 1965 \cite{Golden} and independently the same year by Thompson \cite{Thompson} for all Hermitian matrices $H$ and $K$, \begin{equation}\label{GT}
	\Tr[e^{H+K}]\leq \Tr[e^{H}e^K].	
\end{equation}
Characterizing the general relationship between these operator functions in a similar manner will be the goal of this paper.

\p Each of the operators functions can be extended as a function of $t\in \R$ as \begin{align}
	&(X,Y,t)\mapsto e^{(1-t)\log(X)+t\log(Y)}, \\
	&(X,Y,t)\mapsto e^{(1-t)\log(X)/2}e^{t\log(Y)}e^{(1-t)\log(X)/2}, \\
	&(X,Y,t)\mapsto X^{1/2}(X^{-1/2}YX^{-1/2})^{t}X^{1/2}=:X\#_t Y.
\end{align}
The third operator is referred to as the \textit{weighted geometric mean} for $t\in[0,1]$, and Hiai and Petz proved in 1993 \cite{HiaiPetz} the famous complement to the Golden-Thompson Inequality for Hermitian matrices $H$ and $K$ \begin{equation}
	\Tr\left[e^{H}\#_t e^{K} \right]\leq \Tr\left[e^{(1-t)H+tK} \right].\label{HP}
\end{equation}
This is a corollary of the more general identity proven in 1994 with log majorization techniques \cite{ANDO1994113} \cite{Araki1990} for all $r\geq 0$, $t\in[0,1]$, and unitarily invariant norms $|||\cdot|||$ \begin{equation}\label{inside}
	\big|\big|\big|\left(e^{rH}\#_t e^{rK} \right)^{1/r} \big|\big|\big|\leq \big|\big|\big| e^{(1-t)H+tK} \big|\big|\big|.
\end{equation}

\p The relationship between $\left(e^{rH}\#_t e^{rK} \right)^{1/r} $ and $e^{(1-t)H}e^{tK}$ has been fully characterized: first in \cite{Kian2019} for the range $t\in [-1, \frac{1}{2})$ then in \cite{Hiai2019} for the range $t\geq 0$, from which negative reciprical inequalities can be deduced. This paper for the first time pulls together the full range of $t$ the comparison of all three operators. We do so by proving a similar complimentary identity:

\begin{theorem}\label{theorem}
Let $H$ and $K$ be Hermitian, and $r>0$. Then 
\begin{equation}\label{theq}
\big|\big|\big|e^{(1-t)H+tK} \big|\big|\big|	\leq	\big|\big|\big|\left(e^{rH}\#_t e^{rK} \right)^{1/r} \big|\big|\big|\qquad\;\;\;\;\; t\leq 0,\;\; t\geq 1.
\end{equation}
for any unitarily invariant norm $|||\cdot|||$.
\end{theorem}
Theorem \ref{theorem} is proven in Section 4. It should be noted that Theorem \ref{theorem} can be deduced in the $t\in [-1, \frac{1}{2})$ in \cite{Kian2019}; here we provide proof for the entire range together.

\p This allows us for the first time to characterize the relationship between the three operator functions for all $t\in\R$:
\begin{theorem}\label{bigpicture}
	For Hermitian matrices $H$ and $K$, and $t\in\R$ and any unitarily invariant norm $|||\cdot|||$, then \begin{align}
		&\big|\big|\big|e^H\#_t e^K  \big|\big|\big|\leq \big|\big|\big|e^{(1-t)H+tK} \big|\big|\big|\leq \big|\big|\big|e^{(1-t)H}e^{tK} \big|\big|\big|\qquad 0\leq t\leq 1\label{th11}  \\	
		&\big|\big|\big|e^{(1-t)H+tK} \big|\big|\big|\leq \big|\big|\big|e^{(1-t)H}e^{tK} \big|\big|\big|\leq \big|\big|\big|e^H\#_t e^K  \big|\big|\big|\qquad1\leq t\leq 2\label{th12}	\\
		&\big|\big|\big|e^{(1-t)H+tK} \big|\big|\big|\leq \big|\big|\big|e^H\#_t e^K  \big|\big|\big|\leq \big|\big|\big|e^{(1-t)H}e^{tK} \big|\big|\big| \label{th13}	\qquad \;\;2\leq t.
	\end{align}
\end{theorem}

\p Note that as the trace is a unitarily invariant norm for positive matrices, that the above inequalities all hold for the traces of the matrices. Complimentary negative inequalities can be found taking $X\#_{1-t}=Y\#_t X$.

\begin{proof}
	Equation \ref{th11} comes from taking the $r=1$ case of Equation \ref{inside}, combined with the Golden-Thompson Inequality. Equation \ref{th12} comes taking the $r=1$ case in Equation \ref{theq} and Equation \ref{thH1} (proven in \cite{Hiai2019} and discussed in Section 4) and the Golden-Thompson Inequality for the $t$ positive case, and using Relationship \ref{switch} to extend to the $t$ negative case. Equation \ref{th13} comes from taking the $r=1$ case in Equations \ref{theq} and Equation \ref{thH2} (also proven in \cite{Hiai2019}  and discussed in Section 4) and the Golden-Thompson Inequality for the $t$ positive case, and once more using Relationship \ref{switch} to extend to the $t$ negative case.
	
	Theorem 1 in the case of the trace norm is illustrated in Figure 1.
	
\end{proof}

\begin{figure*}
	\includegraphics[width=\textwidth]{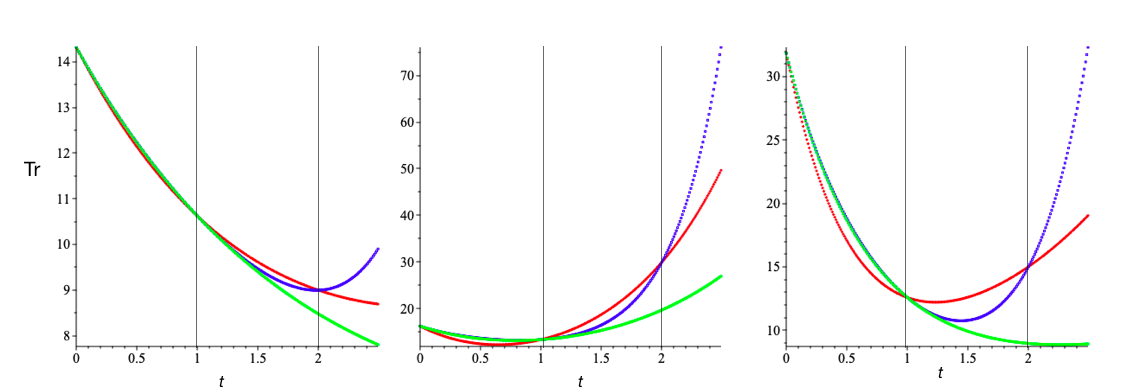}
	\caption{Comparison of 	$\Tr[e^H\#_t e^K ]$ in red, $\Tr[e^{(1-t)H+tK}]$ in green, and $\Tr[e^{(1-t)H}e^{tK}]$ in blue calculated numerically for three sets of randomly generated positive matrices $H$ and $K$.}
	\label{fig:1}      
\end{figure*}

\p Finally, we characterize the equality cases for all of the inequalities in Theorems 1 and 2:
\begin{theorem}\label{theqcase}
	For any strictly increasing unitarily invariant norm $|||\cdot|||$, there is equality in any of the inequalities of Equations \ref{th11}, \ref{th12}, or \ref{th13} if and only if $H$ and $K$ commute. Furthermore, for any strictly increasing unitarily invariant norm $|||\cdot|||$, there is equality in Equation \ref{theq} if and only if $H$ and $K$ commute.
\end{theorem}

\p Section 2 introduces the geodesic interpretation of the weighted geometric mean operator, which gives more significance and clarity to its various properties. Section 3 introduces the technique of log majorization and its relevant applications, which instrumental in prior work in this subject and in proving our main theorems. Section 4 contains the proof of Theorems \ref{theorem} and \ref{theqcase}, commentary on its connection to quantum relative entropy, and the interconnectedness of the $t\in[0,1]$, $t\in[-1,0]\cup[1,2]$, and $t\in(-\infty,-1]\cup[2,\infty)$ cases. 

\section{The Weighted Geometric Mean As Geodesics}

\p Let $P_n\subseteq M_n$ denote the space of $n\times n$ positive definite complex matrices. We consider the Riemannian metric with arc length of the smooth path $\gamma: [a,b]\rightarrow P_n$ defined by \begin{equation}
	L(\gamma):=\int_a^b ||\gamma(t)^{-1/2}\gamma'(t)\gamma(t)^{-1/2}||_2 dt,
\end{equation}
where $||\cdot||_2$ denotes the Hilbert-Schmidt norm. Then the corresponding distance is \begin{equation}
	\delta(X,Y)=\inf\left\{\int_0^1||\gamma(t)^{-1/2}\gamma'(t)\gamma(t)^{-1/2}||_2 dt, \;\;\; \gamma(0)=X,\; \gamma(1)=Y \right\}.
\end{equation}

\p This metric, first introduced by Skovgaard \cite{skovgaard} for its applications in statistics, has the particular nice property that it is invariant under conjugation: \begin{equation}
	\delta(A^\ast XA, A^\ast Y A)=\delta(X,Y)	
\end{equation}
for all invertible matrices $A\in M_n$. It is shown in \cite{Carlen2018} \cite{BhatiaHolbrook} that there is a unique constant speed geodesic for any $X,Y\in P_n$ running between $X$ and $Y$ in unit time, namely \begin{equation}
	\gamma(t)=X\#_t Y:=X^{1/2}(X^{-1/2}YX^{-1/2})^tX^{1/2}.
\end{equation}

\p These geodesics satisfy \cite{Carlen2018} for all $t,t_0,t_1\in\R$ and $X,Y\in P_n$  \begin{equation}
	X\#_{(1-t)t_0+tt_1}Y =(X\#_{t_0}Y)\#_t(X\#_{t_1}Y).
\end{equation}
Of particular interest to us will be the $t_0=1$, $t_1=0$, which gives the identity \begin{equation}\label{switch}
	X\#_{1-t}Y=Y\#_t X.
\end{equation}	This will allow us to extend all inequalities proven for $t\in\R^+$ to $(1-t)\in\R^-$, and in particular to connect the positive and negative cases of Theorem \ref{bigpicture}.

\p Ando and Kubo prove in \cite{Kubo1980} that the map \begin{equation}
	(X, Y ) \mapsto X\#_tY 	
\end{equation}
is jointly concave and monotone increasing in X and Y for $t\in[0,1]$. Carlen and Lieb prove \cite{Carlen2018} that it is jointly convex for $t \in [-1, 0] \cup [1, 2]$. However, it is challenging to expand these results as $f(x)=x^t$ is not operator monotone for $t>1$ and neither concave nor convex for $t\notin [-1, 2]$. Therefore, proofs of trace inequalities cannot rely solely on operator monotonicty or convexity, so we turn to our main tool: log majorization. 

\section{Log Majorization}
\p Let $\mathbf{a}=(a_1,\dots, a_n)$ and $\mathbf{b}=(b_1,\dots, b_n)$, with $a_1\geq\dots\geq a_n$ and $b_1\geq\dots\geq b_n$. Then $\mathbf{b}$ weakly majorizes $\mathbf{a}$, written $\mathbf{a}\prec_{w} \bf{b}$, when \begin{equation}
	\sum_{i=1}^ka_i\leq \sum_{i=1}^k b_i, \qquad 1\leq k \leq n
\end{equation}
and is majorized $\bf{a}\prec\bf{b}$ when the final inequality is an equality. Weak log majorization $\bf{a}\prec_{w(\log)}\bf{b}$ is similarly defined for non-negative vectors as \begin{equation}
	\prod_{i=1}^ka_i\leq \prod_{i=1}^k b_i, \qquad 1\leq k \leq n
\end{equation}
with log majorization $\bf{a}\prec_{(\log)}\bf{b}$ when the final inequality is an equality.

\p We define all of the above majorization for matrices, ie $A\prec B$ and all variations, when the singular values in descending order considered as a vector $(s_1(A), \dots, s_n(A )) \prec (s_1(B), \dots, s_n(B))$.

\p We list the following results for convenience: 

\begin{lemma}\label{convlemma}
	Let $A,B\in M_n^+$. Then the following conditions are equivalent: $A\prec_{w(log)}B$, $|||A|||\leq |||B|||$ for every unitarily invariant norm $|||\cdot|||$, and $|||f(A)|||\leq |||f(B)|||$ for every unitarily invariant norm $|||\cdot|||$ and continuous increasing function $f:\R\rightarrow\R$ such that $f(0)\geq 0$ and $f(e^t)$ is convex. 
\end{lemma}

\p A full proof can be found in \cite{Hiai2014}. 
\begin{corollary}\label{flog}
	When $A\prec_{(\log)}B$, and for a function $f$ that satisfies the hypotheses of Lemma \ref{convlemma} such that $f(x_1x_2)=f(x_1)f(x_2)$, then $f(A)\prec_{(\log)}f(B)$.
\end{corollary}
\begin{proof}
	Lemma 1 implies $f(A)\prec_{w(\log)}f(B)$, then $\prod_{i=1}^n \lambda(A)=\prod_{i=1}^n \lambda(B)$ gives $\prod_{i=1}^n f(\lambda(A))=f(\prod_{i=1}^n \lambda(A))=f(\prod_{i=1}^n \lambda(B))=\prod_{i=1}^n f(\lambda(B))$. 
\end{proof}

\begin{theorem}\label{gtlog} (Araki \cite{Araki1990}, Theorem 1)
	For all $A,B\in P_n$,	 \begin{equation}
		(A^{1/2} B A^{1/2})^r \prec_{(\log)} A^{r/2} B^r A^{r/2}, \qquad (r \geq 1)	 
	\end{equation}
	or equivalently \begin{equation}
		(A^{p/2}B^p A^{p/2})^{1/p} \prec_{(\log)}  (A^{q/2}B^q A^{q/2})^{1/q}, \qquad (0 < p \leq q).	
	\end{equation}
\end{theorem}

\begin{corollary}\label{6.39}
	For all self-adjoint $H,K\in M_n$ and unitarily invariant norm  $|||\cdot |||$, \begin{equation}
		\big|\big|\big|e^{H+K} \big|\big|\big| \leq \big|\big|\big|(e^{q H/2}e^{q K} e^{q H/2})^{1/q} \big|\big|\big|, \qquad\qquad\qquad q > 0,
	\end{equation}
	or equivalently \begin{equation}
		e^{H+K} \prec_{(\log)}(e^{q H/2}e^{q K} e^{q H/2})^{1/q}, \;\;\;\qquad\qquad\qquad\qquad q > 0,
	\end{equation}
\end{corollary}
\begin{proof}
	Take $p\rightarrow 0$ in Theorem \ref{gtlog}  and apply the Lie-Trotter product formula. 
\end{proof}

\section{Trace Inequalities}

\begin{proof}[Proof of Theorem \ref{theorem}]
	Let $t\geq 1$, and $r>0$. For $f(x)=x^{1/r}$, then $f(0)=0$, $f(e^t)$ is convex and increasing, and $f(x_1x_2)=f(x_1)f(x_2)$. Then applying Corollary \ref{flog} and Theorem \ref{gtlog}, \begin{align}
		(e^{rH}\#_t e^{rK})^{1/r}&=f(e^{rH/2}(e^{-rH/2}e^{rK}e^{-rH/2})^te^{rH/2}) \\
		&\succ_{(\log)} f\left((e^{rH/2t}(e^{-rH/2}e^{rK}e^{-rH/2})e^{rH/2t})^t\right) \\
		&=(e^{r(1-t)H/2t}e^{rK}e^{r(1-t)H/2t})^{t/r}.
	\end{align}
	We explicitly write\begin{equation}
		e^{r(1-t)H/2t}e^{rK}e^{r(1-t)H/2t}=e^{r(1-t)H/2t}e^{trK/t}e^{r(1-t)H/2t}
	\end{equation}As $t/r>0$, we apply Corollary \ref{6.39}: \begin{equation}
		(e^{r(1-t)H/2t}e^{trK/t}e^{r(1-t)H/2t})^{t/r} \succ_{(\log)} e^{(1-t)H+tK}.
	\end{equation} 
	Then by Lemma 1,\begin{equation}\label{theq2}
		|||(e^{rH}\#_t e^{rK})^{1/r}|||\geq |||e^{(1-t)H+tK}|||
	\end{equation}
	for any unitarily invariant norm. Relationship \ref{switch} extends this to $t\leq 0$. 
\end{proof}

\begin{proof}[Proof of Theorem \ref{theqcase}]
	When $H$ and $K$ commute, all of the expressions being considered are equal, and hence their norms will be equal. Therefore, it remains to consider the implications of equality for strictly increasing unitarily invariant norms.
	
\p 	It is known from Hiai \cite{GTEqualityCases} \textit{(Theorem 3.1)} that there is equality in the inequalities in Equation \ref{th11} for a strictly increasing unitarily invariant norm if and only if $H$ and $K$ commute, and in \cite{Hiai2019} for all the relationships between $|||(e^{rH}\#_te^{rK})^{1/t}|||$ and $|||e^{(1-t)H}e^{tK}|||$.  Therefore, it remains to prove that there is equality in Equation \ref{theq} (re-stated above in Equation \ref{theq2}) if and only if $H$ and $K$ commute; then Relationship \ref{switch} extends the characterization to all $t$, taking $r=1$.
	
\p 	We use the following Lemma from Hiai \cite{GTEqualityCases}:
	
	\begin{lemma}\label{eqcase} (Hiai 1994) Let $A,B\in M_n^+$. Then $|||(A^{p/2}B^pA^{p/2})^{1/p}|||$ is not strictly increasing in $p>0$ if and only if $A$ and $B$ commute.
	\end{lemma}
	
\p 	From the proof of Theorem \ref{theorem}, we have \begin{equation}
		(e^{rH}\#_t e^{rK})^{1/r}\succ_{{\log}}	(e^{r(1-t)H/2t}e^{trK/t}e^{r(1-t)H/2t})^{t/r} \succ_{(\log)} e^{(1-t)H+tK}.
	\end{equation}
	Then for a strictly increasing unitarily invariant norm, it follows that we have \begin{equation}\label{ueq}
		\big|\big|\big|(e^{rH}\#_t e^{rK})^{1/r}\big|\big|\big|=\big|\big|\big|(e^{r(1-t)H/2t}e^{trK/t}e^{r(1-t)H/2t})^{t/r}\big|\big|\big| = \big|\big|\big|e^{(1-t)H+tK}\big|\big|\big|.
	\end{equation}
	We then apply Lemma \ref{eqcase} to the second equality considering the Lie-Trotter product formula. 
\end{proof}

\p It is interesting to note that Equation \ref{theq}, the $t\in (-\infty, 0]\cup [1, \infty)$ case actually implies Equation \ref{inside}, the $t\in[0,1]$ case. In \cite{HiaiPetz}, Hiai and Petz show that Equation \ref{inside} is equivalent to  \begin{equation}\label{dereq}
	\frac{1}{r}\Tr\left[e^H\log(e^{rH/2}e^{-rK}e^{rH/2}) \right]\geq \Tr\left[e^{H}(H-K) \right].
\end{equation}
for all $r>0, t\in[0, 1]$. The forward implication comes from noting the inequality trivially becomes equality $\Tr[e^{H}]=\Tr[e^{H}]$ at $t=0$, so taking the limit from above one can relate the derivatives of each side. The backwards inequality is far more involved. Hiai and Petz prove Equation \ref{inside} by proving Equation \ref{dereq} then making use of their equivalence.

\p The exact same reasoning behind the forward implication applies substituting in knowledge of the trace inequality for $t\leq0$. Therefore, it is an immediate consequence of \cite{HiaiPetz} that Equation \ref{theq} implies Equation \ref{inside}: taking the derivative of Equation \ref{theq} at $t=0$ with some changes of variable produces Equation \ref{dereq} directly from Equation \ref{theq}.

\p As mentioned in Section 1, Hiai \cite{Hiai2019} recently published work regarding log majorization and R\'{e}nyi divergences, including that for $t\geq 1$,
\begin{align}
	e^{(1-t)H/2}e^{tK}e^{(1-t)H/2}  \prec_{(\log)}(e^{rH}\#_t e^{rK})^{1/r}  \; \;\;\; r\geq \max\{t/2,t-1 \}\label{lthH1} \\
	(e^{rH}\#_t e^{rK})^{1/r}   \prec_{(\log)} e^{(1-t)H/2}e^{tK}e^{(1-t)H/2} 	\;\;\;\; r\leq \min\{t/2,t-1 \}.\label{lthH2}
\end{align}

\p It then follows from Lemma \ref{convlemma} that for $t\geq 1$,
\begin{align}
	\big|\big|\big|e^{(1-t)H/2}e^{tK}e^{(1-t)H/2}\big|\big|\big| \leq \big|\big|\big|(e^{rH}\#_t e^{rK})^{1/r}\big|\big|\big|  \;\;\;\; r\geq \max\{t/2,t-1 \}\label{thH1} \\
	\big|\big|\big|(e^{rH}\#_t e^{rK})^{1/r}\big|\big|\big|   \leq \big|\big|\big|e^{(1-t)H/2}e^{tK}e^{(1-t)H/2} \big|\big|\big|	\;\; \;\; r\leq \min\{t/2,t-1 \}.\label{thH2}
\end{align}
for any unitarily invariant norm $|||\cdot|||$.

\p It is possible to prove a simplified version of Equation \ref{lthH1} (the special case of r=1) using similar methods as in \cite{Hiai2019}, with the Furuta inequality:

\begin{theorem}\label{mylog1}
	Let $H,K\in M_n$ be Hermitian, and $1\leq t\leq 2$. Then \begin{equation}
		e^{(1-t)H/2}e^{tK}e^{(1-t)H/2}\prec_{(\text{log})} e^{H}\#_t e^{K}.	
	\end{equation}
\end{theorem}
\begin{proof} We use the standard antisymmetric tensor power technique to prove log majorization, noting that for $A \in M_n$ and $k = 1,\dots,n,$ then \begin{equation}
		\prod_{i=1}^k s_i(A)=s_1(A^{\wedge k})=||A^{\wedge k}||,
	\end{equation}
	and that the antisymmetric tensor product has the properties $(A^{\wedge k})^\ast=(A^\ast)^{\wedge k}$, $(AB)^{\wedge k}=A^{\wedge k}B^{\wedge k}$, and for all $A\geq 0$ and $r>0$, $(A^r)^{\wedge k}=(A^{\wedge k})^r$ \cite{Hiai20142}. Then as \begin{equation}
		\det(e^{(1-t)H/2}e^{tK}e^{(1-t)H/2})=\det(e^H)^{1-t}\det(e^K)^t=\det(e^{H}\#_t e^{K}),
	\end{equation}
	it suffices to show that \begin{equation}
		||e^{(1-t)H/2}e^{tK}e^{(1-t)H/2}||\leq ||e^{H}\#_t e^{K}||.		
	\end{equation}
	
\p	Suppose $e^{H}\#_t e^{K}	\leq I$. We use the following relation \cite{FURUTA1995139} for all $A\in P_n$, $B\in M_n$ invertible, and real number $r$: \begin{equation}
		(BAB^\ast)^r=BA^{1/2}(A^{1/2}B^\ast BA^{1/2})^{r-1}A^{1/2}B^\ast.	
	\end{equation}
	Then \begin{align}
		&e^{H/2}(e^{-H/2}e^Ke^{-H/2})^t e^{H/2}\leq I \\
		&\qquad \Rightarrow (e^{-H/2}e^Ke^{-H/2})^t \leq e^{-H} \\
		&\qquad \Rightarrow e^{-H/2}e^{K/2}(e^{K/2}e^{-H}e^{K/2})^{t-1}e^{K/2}e^{-H/2} \leq e^{-H} \\
		&\qquad \Rightarrow (e^{K/2}e^{-H}e^{K/2})^{t-1}\leq e^{-K}.
	\end{align}

\p	We now apply the Furuta Inequality \cite{Furuta1987} for $A\geq B\geq 0$, $r\geq 0$, $p\geq 0$, $q\geq 1$ with $(1+2r)q\geq p+2r$: \begin{equation}
		(A^r B^p A^r)^{1/q}\leq A^{(p+2r)/q}.
	\end{equation}
	We choose $A=e^{-K}$, $B=(e^{K/2}e^{-H}e^{K/2})^{t-1}$, $r=\frac{1}{2}$, $p=(t-1)^{-1}$, $q=(t-1)^{-1}$. Note that as $t\in[1,2]$ then $t-1\in[0,1]$, and $q\geq 1$. Furthermore, \begin{equation}
		(1+2r)q-(p+2r)=\frac{2}{t-1}-\frac{1}{t-1}-1=\frac{1}{t-1}-1\geq 0,
	\end{equation}
	so all the hypotheses of the inequality are satisfied. Then \begin{equation}
		A^{(p+2r)/q}=e^{-(\frac{1}{t-1}+1)(t-1)K}=e^{-tK},
	\end{equation}
	and \begin{align}
		(A^r B^p A^r)^{1/q}=(e^{-K/2}((e^{K/2}e^{-H}e^{K/2})^{t-1})^{1/(t-1)}e^{-K/2})^{t-1}=e^{-(t-1)H}.
	\end{align}
	Therefore, we conclude \begin{equation}
		e^{-(t-1)H}\leq e^{-tK},	
	\end{equation}
	which implies that \begin{equation}
		e^{tK}\leq e^{(t-1)H}	
	\end{equation}
	and so \begin{equation}
		e^{(1-t)H/2}e^{tK}e^{(1-t)H/2}\leq I.		
	\end{equation}
\end{proof}

\p Hiai provides a full analysis on the connection to R\'{e}nyi divergences by noting that in the $t=2$ case, $e^H\#_te^K=e^{(1-t)H}e^{tK}$, and one can use the same technique taking the derivative:
\begin{align}
	\frac{d}{dt}	\Tr\left[e^{H}\#_t e^{K}  \right]\Big|_{t=2}=
	&=\Tr\left[e^{H/2}\log(e^{-H/2}e^{K}e^{-H/2})(e^{-H/2}e^{K}e^{-H/2})^2e^{H/2} \right]\\
	&=\Tr\left[e^{H/2}e^{-H/2}e^{K}e^{-H/2}\log(e^{-H/2}e^{K}e^{-H/2})e^{-H/2}e^{K}e^{-H/2}e^{H/2} \right]\\
	&=\Tr\left[e^{K}e^{-H/2}\log(e^{-H/2}e^{K}e^{-H/2})e^{-H/2}e^{K} \right]\\
	&=\Tr\left[e^{-H/2}e^{2K}e^{-H/2}\log(e^{-H/2}e^{K}e^{-H/2}) \right]
\end{align}
which using 
\begin{equation}
	\frac{d}{dt}\Tr\left[e^{(1-t)H}e^{tK} \right] \Big|_{t=2}=\Tr\left[(K-H)e^{-H}e^{2K}\right]
\end{equation}
gives \begin{equation}
	\Tr\left[e^{-H/2}e^{2K}e^{-H/2}\log(e^{-H/2}e^{K}e^{-H/2}) \right]\leq \Tr\left[e^{-H}e^{2K}(K-H)\right]
\end{equation}

\p Theorem \ref{mylog1} gives the same derivative identity at $t=2$. We thus conjecture that similar methods to \cite{HiaiPetz} it may be possible to show that \begin{equation}
	\Tr\left[e^{(1-t)H/2}e^{tK}e^{(1-t)H/2} \right]\leq \Tr\left[e^{H}\#_t e^{K} \right]\qquad\qquad\qquad 1\leq t\leq 2
\end{equation}
also implies 
\begin{equation}
	\Tr\left[e^{H}\#_t e^{K} \right]\leq \Tr\left[e^{(1-t)H/2}e^{tK}e^{(1-t)H/2} \right]\qquad\qquad\qquad 2\leq t.
\end{equation}

\section{Acknowledgements}
Thank you to Professor Eric A Carlen for suggesting the problem and providing an introduction to the general subject. 

This research was funded by the NDSEG Fellowship, Class of 2017.

\bibliographystyle{spmpsci}   
\bibliography{references} 
\end{document}